\documentclass[a4paper]{article}
\usepackage[utf8]{inputenc}
\usepackage[a4paper, total={16cm, 22cm}]{geometry}

\usepackage[affil-it]{authblk}

\usepackage{amsmath,amsthm,amssymb,amsfonts,stmaryrd,tensor,scalerel}
\usepackage{graphicx}
\usepackage{xcolor}
\usepackage[all,cmtip]{xy}
\usepackage{placeins}
\usepackage{array}
\usepackage{multirow}
\usepackage{comment}

\newtheorem{proposition}{Proposition}[section]

\newtheorem{lemma}[proposition]{Lemma}

\theoremstyle{plain}
\newtheorem{remark}{Remark}[section]
\theoremstyle{definition}

\newcommand{\de}{\,\mathrm{d}}

\newcommand{\rn}[1]{{\mathbb{R}^#1}}
\newcommand{\x}{\mathbf{x}}
\newcommand{\X}{\mathbf{X}}

\newcommand{\cotb}[1]{T^*\mathbb{R}^#1}

\newcommand{\MA}{Monge-Ampère }
\newcommand{\MK}{Monge-Kantorovich }
\newcommand{\LR}{Lychagin-Rubtsov }
\newcommand{\KMW}{Kim-McCann-Warren }

\newcommand{\domain}{M}
\newcommand{\codomain}{{\bar M}}
\newcommand{\symplectic}{\omega}
\newcommand{\effective}{\alpha}

\title{A note on optimal transport and \MA geometry}
\author{R. D'Onofrio
\footnote{\textit{Dipartimento di Matematica e Applicazioni, Universit\`a di  Milano-Bicocca, Via Roberto Cozzi 55, I-20125 Milano, r.donofrio1@campus.unimib.it\\
Department of Mathematics, University of Surrey, Guildford GU2 7XH, UK, r.d'onofrio@surrey.ac.uk\\
INFN, Sezione di Milano-Bicocca, Piazza della Scienza 3, 20126 Milano, Italy}}
}

\date{\today}

\begin{document}

\maketitle

\begin{abstract}
We identify a novel connection between a recently introduced pseudo-Riemannian framework for optimal mass transport and the geometry of \MA equations.
We show this correspondence by application to an example from geophysical fluid dynamics.
\end{abstract}

\section{Introduction}
In recent years, differential geometry has emerged as a key aspect in both optimal mass transport \cite{Del2008,KMW2010} and \MA equations \cite{Ban2002,KLR2006}.
In particular, Kim, McCann and Warren \cite{KMW2010} have applied ideas from pseudo-Riemannian geometry to investigate the geometrical properties of optimal maps, while Lychagin and Roubtsov \cite{LR83} 
have used pseudo-Riemannian geometry to classify symplectic \MA equations in three independent variables.

In this short account, we identify a close correspondence between these two mathematical frameworks in the guise of a conformal relation between two pseudo-Riemannian metrics. Moreover, we discuss an example from geophysical fluid dynamics which motivates our assumptions.

\section{Pseudo-Riemannian geometry in optimal mass transport and \MA equations}


Let us start by reviewing the fundamental relations between optimal mass transport and \MA equations (see for example \cite{KMW2010} and the literature cited therein).
Given two $n$-dimensional manifolds $\domain$ and $\codomain$, equipped with probability densities $\rho(x)$ and $\bar \rho(\bar x)$, and a smooth function $c(x,\bar x):\domain\times \codomain\to\mathbb R$, the \MK problem is to find a joint probability density $\gamma(x,\bar x)$ (called the transportation plan) which minimises the total cost,
\begin{equation}
\label{MK problem}
	\mathcal I[\gamma]=\int_{\domain\times \codomain} c(x,\bar x)\gamma(x,\bar x)d x d \bar x,
\end{equation}
and conserves mass. The conservation of mass is expressed by the constraints 
\begin{equation}
\label{marginals}
\int_\domain \gamma(x,\bar x)d x=\bar \rho(\bar x) \quad \textnormal{and} \quad \int_\codomain \gamma(x,\bar x)d \bar x=\rho(x).
\end{equation}
The \MK problem has a dual formulation which is to find a pair of potentials $(u(x),\bar u(\bar x))$ which maximise
\begin{equation}
\label{dual MK problem}
	\mathcal J[u,\bar u]=-\int_\domain u(x)\rho(x)dx-\int_\codomain \bar u(\bar x)\bar\rho(\bar x)d\bar x,
\end{equation}
subject to
\begin{equation}
\label{constraint potentials}
	u(x)+\bar u(\bar x)\ge -c(x,\bar x).
\end{equation}
The two formulations of the \MK problem are dual of each other in the sense that
\begin{equation}
\inf_{\gamma\textnormal{ satisfying }(\ref{marginals})}\mathcal I[\gamma]=\sup_{u,\bar u\textnormal{ satisfying }(\ref{constraint potentials})} \mathcal J[u,\bar u].
\end{equation}
Under standard hypothesis on the cost function (hypothesis A1 and A2 of \cite{MTW2005}), it can be proved that a maximising pair $(u,\bar u)$ always exists and is unique up to a constant. The same hypothesis also imply that the optimal measure $\gamma$ is supported on the graph of a map $T_u:\domain\to\codomain$ called the transportation map.
The optimal potential $u(x)$ and the transportation map $T_u(x)$ are found from the Euler-Lagrange system
\begin{equation}
	\begin{cases}
	\label{E-L}
	Du(x)=-D_x c(x,T_u(x)),\\
	\bar\rho(T_u(x))\det(D^2u(x)+D_x^2c(x,T_u(x)))=\rho(x)\det(-D_xD_{\bar x}c(x,T_u(x))).
	\end{cases}
\end{equation}
The optimal map is always invertible, and its inverse $T_{\bar u}(\bar x)$ along with the optimal potential $\bar u(\bar x)$ solve a system of equations analogous to (\ref{E-L}).

\begin{remark}
When $M=\bar M=\rn{n}$, an important particular case is represented by the cost function
\begin{equation}
\label{simple cost}
	c(x,\bar x)=-x\cdot\bar x,
\end{equation}
for which equations (\ref{E-L}) take the form
\begin{equation}
\begin{cases}
\label{simple E-L}
	\bar x=Du(x),\\
	\bar\rho(\bar x)\det(D^2u)=\rho(x).
\end{cases}
\end{equation}
\end{remark}

\begin{remark}
The optimal potentials satisfy
\begin{equation}
\label{c-duality}
u(x)+\bar u(\bar x)=-c(x,\bar x).
\end{equation}
Equation (\ref{c-duality}) can be interpreted as a generalised duality relation between the optimal potentials $u(x)$ and $\bar u(\bar x)$ in the terminology of Sewell \cite{Sew2002}.
For the cost (\ref{simple cost}), equation (\ref{c-duality}) takes the familiar form $u(x)+\bar u(\bar x)=x\cdot{\bar x}$ and the optimal potentials are Legendre dual of each other.
\end{remark}

\subsection{Geometry of optimal transport}

With the aim of characterising the geometrical properties of optimal maps, Kim \textit{et al.} \cite{KMW2010} equipped the product $\domain\times\codomain$ with a symplectic form,
\begin{equation}
\label{omega c}
	\omega_c=\begin{pmatrix}0 & -D_xD_{\bar x}c\\ (D_xD_{\bar x}c)^T & 0\end{pmatrix},
\end{equation}
and a pseudo-Riemannian metric of type $(n,n)$,
\begin{equation}
\label{KMW metric}
	h_c^{\rho,\bar \rho}=\left(\frac{\rho(x)\bar\rho(\bar x)}{|\det(D_xD_{\bar x}c)|}\right)^{1/n}
	\begin{pmatrix}
	0 & -D_xD_{\bar x}c\\
	-(D_xD_{\bar x}c)^T & 0
	\end{pmatrix}.
\end{equation}
Then, they were able to show that the graph of the optimal map is a volume-maximising space-like Lagrangian submanifold of $\domain\times \codomain$.

\subsection{Geometry of \MA equations}

Geometry of symplectic \MA equations\footnote{The term symplectic here means that the equation's coefficients can only depend on the independent variables and the first derivatives of the dependent variable.} deals with the cotangent bundle to the manifold of independent variables equipped with a \textit{\MA structure}.
For equations (\ref{E-L}) or (\ref{simple E-L}), the manifold of independent variables is naturally identified with $M$.
 A \MA structure on $T^*\domain$ is a pair of differential forms  $(\symplectic,\effective)$ where $\symplectic$ is a symplectic 2-form and $\effective$ is an \textit{effective} $n$-form, that is, $\symplectic\wedge\effective=0$.
Each effective $n$-form $\effective$ on $T^*\domain$ identifies a unique symplectic \MA operator $\Delta_\effective:C^\infty(\domain)\to\Omega^n(\codomain)$, taking functions to volume forms on $\domain$, by
 \begin{equation}
 	\Delta_\effective(u):=(du)^*\effective.
 \end{equation}
The \MA equation corresponding to $\effective$ is $\Delta_\effective(u)=0$. Every smooth function $u\in \ker(\Delta_\effective)$ is called a classical solution of the \MA equation. A (generalised) solution is a Lagrangian submanifold $L\subset T^*\domain$ such that $\effective|_L=0$. 

The \MA structure induces a pseudo-Riemannian metric on $T^*M$. For $n=3$, it can be computed through the formula
\begin{equation}
\label{LR metric 3d}
    g_\effective(X_1,X_2)\frac{\symplectic\wedge\symplectic\wedge\symplectic}{3!}=\iota_{X_1} \effective\wedge \iota_{X_2}\effective\wedge \symplectic,
\end{equation}
where $\iota$ denotes the interior product and $X_1,X_2$ are two vector fields on $\cotb{3}$ (see \cite{BRR2016}).
This object was first introduced by Lychagin and Roubtsov \cite{LR83} in order to classify symplectic \MA equations in three independent variables, and later generalised to any dimension $n$ (see for example \cite{Rub2019}).

\subsection{Main results}
In order to compare the geometry of optimal transport and the geometry of \MA equations, we make the identification $T^*\domain\simeq \domain\times\codomain$. This is certainly possible if $M=\bar M=\rn{n}$, and at least locally for more general manifolds. Moreover, we equip this manifold with the symplectic form $\omega_c$ (cf. equation (\ref{omega c})). 
Then, we show that, for $c(x,\bar x)$ given by (\ref{simple cost}), the \KMW metric and the \LR metric are conformally equivalent.

\begin{proposition}
Let $M=\bar M=\rn{3}$ and the cost be given by (\ref{simple cost}). Then, the following conformal relation holds
\begin{equation}
\label{conformal relation}
    g_{\effective}=[\rho(x)\bar\rho(\bar x)]^{2/3}h_c^{\rho,\bar\rho}.
\end{equation}
\end{proposition}

\begin{proof}[Proof]
First observe that the effective form $\alpha$ corresponding to the \MA equation $(\ref{simple E-L})_2$ is, in coordinates $(x,\bar x)$,
\begin{equation}
	\effective=\bar\rho(\bar x)d\bar x^1\wedge\de\bar x^2\wedge d\bar x^3-\rho(x)d x^1\wedge d x^2\wedge d x^3.
\end{equation}
Therefore, formula (\ref{LR metric 3d}) yields, by also setting $\omega=\omega_c=dx\wedge d\bar x$,
\begin{equation}
\label{LR metric}
	g_\effective=2\rho(x)\bar\rho(\bar x)d x d\bar x.
\end{equation}
On the other hand, using (\ref{simple cost}) in equation (\ref{KMW metric}) gives
\begin{equation}
	h_c^{\rho,\bar\rho}=2(\rho(x)\bar\rho(\bar x))^{1/3}d x d\bar x.
\end{equation}
Then, (\ref{conformal relation}) follows.
\end{proof}

\begin{remark}
Under the hypothesis of Proposition 2.1, $(x,\bar x)$ are canonical coordinates on $T^*\domain\simeq \domain\times\codomain$, that is, $\omega_c=dx\wedge d\bar x$. This is generally not true for more general cost functions.
\end{remark}

\section{Example: semigeostrophic theory}
Several applications of optimal transport in semigeostrophic theory are lucidly presented in \cite{Cul2021}. Here, we summarise the main facts about $f$-plane semigeostrophic theory as an example where the results of Section 2.3 apply.

In the $f$-plane approximation, the state of each fluid particle is characterised by its position $\x=(x,y,z)$ and a triple $\X=(X,Y,Z)$, where $X$ and $Y$ represent a normalised expression for the horizontal absolute momentum and $Z$ represents a normalised potential temperature. The total energy of the fluid occupying a domain $M\subset\rn{3}$ at some fixed time $t$ is expressible as
\begin{equation}
\label{energy bis}
    E=f^2\int_\domain\bigg[\frac{(x-X)^2}{2}+\frac{(y-Y)^2}{2}-zZ\bigg]\rho(\x)d\x,
\end{equation}
where $f$ is a physical constant known as the Coriolis parameter and $\rho(\x)$ represents the fluid density.
The set of all the $\X$-values of particles at time $t$ form another subdomain of $\rn{3}$, which we denote $\codomain$. 

Semigeostrophic flows obey a minimum energy principle \cite{Cul2021}, which can be cast as a \MK problem where the cost function is the integrand in (\ref{energy bis}) and $\rho(\x)$ and $\bar \rho(\X)$ are some given probability densities on $\domain$ and $\codomain$ respectively. 
Equations (\ref{E-L}) in this case take the form
\begin{equation}
\label{E-L SG}
	\begin{cases}D u=\X-(x,y,0),\\ \bar\rho(\X)\det(D^2u+D^2_xc)=\rho(\x), \end{cases}
\end{equation}
where $D_x^2c$ is a constant diagonal matrix with entries $(1,1,0)$. Although equations (\ref{E-L SG}) may look different from (\ref{simple E-L}), the results of Proposition 2.1 still apply. More precisely,

\begin{proposition}
Let the cost be given by the integrand in (\ref{energy bis}). Then, $(\x,\X)$ are canonical coordinates on $T^*\domain\simeq \domain\times\codomain$ equipped with the symplectic form (\ref{omega c}), and the conformal relation (\ref{conformal relation}) holds.
\end{proposition}

\begin{proof}
The first point easily follows from (\ref{omega c}) by setting $c=(x-X)^2/2+(y-Y)^2/2-zZ$. For the second point, observe that the effective 3-form corresponding to the \MA equation $(\ref{E-L SG})_2$ is
\begin{equation}
\alpha=\bar \rho(\X)dX\wedge dY\wedge dZ-\rho(\x)dx\wedge dy\wedge dz.
\end{equation}
Then, formula (\ref{LR metric 3d}) yields, by also setting $\omega=\omega_c=d\x\wedge d\X$,
\begin{equation}
g_\alpha=2\rho(\x)\bar \rho(\X)(dxdX+dydY+dzdZ).
\end{equation}
On the other hand, equation (\ref{KMW metric}) gives
\begin{equation}
h_c^{\rho,\bar\rho}=2\rho(\x)^{1/3}\bar\rho(\X)^{1/3}(dxdX+dydY+dzdZ).
\end{equation}
Therefore, equation (\ref{conformal relation}) is verified.
\end{proof}

\section{Conclusions and future directions}
We have shown a conformal relation between the \KMW metric and the \LR metric in dimension $n=3$, when the transportation cost is (\ref{simple cost}). Relaxing any of these hypothesis leads to possible directions for future investigation. In particular, the work of Cullen \textit{et al.} on semigeostrophic theory on a sphere is an example where the cost function takes a form different from (\ref{simple cost}).

\section*{Acknowledgements}
The author would like to thank G. Ortenzi, I. Roulstone and V. Roubtsov for their constant support and useful advises. The author is also grateful to T. Bridges, L. Napper and M. Wolf for useful discussions. A special thanks goes to M. Cullen, who introduced the author to the \MK problem in semigeostrophic theory.
This work was supported by the European Union’s Horizon 2020 research and innovation program under the Marie Skłodowska-Curie grant no 778010 IPaDEGAN. The author thanks the financial support of the MMNLP (Mathematical Methods in Non Linear Physics) project of the INFN.
The author also gratefully acknowledge the auspices of the GNFM Section of INdAM under which part of this work was carried out. This work is part of the author's dual PhD program Bicocca-Surrey.



\end{document}